\documentclass[times,12pt]{rncauth}

\usepackage{graphicx,amssymb,amstext,amsmath}
\usepackage{verbatim}
\usepackage{cite}
\usepackage[dvips,colorlinks,bookmarksopen,bookmarksnumbered,citecolor=red, urlcolor=red]{hyperref}

\def\volumeyear{2016}

\newtheorem{theorem}{Theorem}[section]

\theoremstyle{definition}
\newtheorem{definition}{Definition}[section]

\theoremstyle{remark}
\newtheorem{remark}{Remark}

\newenvironment{assumption}[1][Assumptions]{\begin{trivlist}
\item[\hskip \labelsep {\bfseries #1}]}{\end{trivlist}}

\begin{document}

\runningheads{S.~Roy and I.~R.~Petersen}{Robust $H_\infty$ Estimation of Uncertain Linear Quantum Systems}

\title{Robust $H_\infty$ Estimation of Uncertain Linear Quantum Systems}

\author{Shibdas Roy\affil{1}\corrauth and Ian~R.~Petersen\affil{2}}

\address{\affilnum{1}Department of Electrical and Computer Engineering, National University of Singapore, Singapore.\break \break \affilnum{2}School of Engineering and Information Technology, University of New South Wales, Canberra, Australia.}

\corraddr{E-mail: roy\_shibdas@yahoo.co.in}

\cgs{Australian Research Council}

\begin{abstract}
We consider classical estimators for a class of physically realizable linear quantum systems. Optimal estimation using a complex Kalman filter for this problem has been previously explored. Here, we study robust $H_\infty$ estimation for uncertain linear quantum systems. The estimation problem is solved by converting it to a suitably scaled $H_\infty$ control problem. The solution is obtained in the form of two algebraic Riccati equations. Relevant examples involving dynamic squeezers are presented to illustrate the efficacy of our method.
\end{abstract}


\maketitle

\section{Introduction}
In recent years, there has been significant interest in studying estimation and control problems for quantum systems \cite{WM1,YK1,YK2,NY,JNP,NJP,GGY,MP1,MP2,YNJP,GJ,GJN,IRP1}. Linear quantum systems are an important class of quantum systems and have been of particular interest in this context \cite{WM1,NY,JNP,NJP,GGY,GJN,WM2,GZ,WD,NJD,HM,SSM,IRP3}. Such linear quantum systems have been useful in describing quantum optical devices such as linear quantum amplifiers \cite{GZ}, finite bandwidth squeezers \cite{GZ} and optical cavities \cite{WM2,BR}. Coherent feedback control has also been studied a lot recently for linear quantum systems \cite{JNP,NJP,MP1,MP2,HM,WM3,SL,GW,IRP3}. The authors have previously explored a related coherent-classical estimation problem \cite{IRP2,RPH}, where the estimator consists of a classical part, that produces the desired final estimate and a quantum part, which may also involve coherent feedback. In this work, the authors have studied optimal classical estimation for linear quantum systems using a complex Kalman filter.

Ref. \cite{MJ} considered a quantum observer, that is a purely quantum system, which produces a quantum estimate of a variable for the quantum plant. By contrast, we here consider classical estimation for linear quantum systems, where the estimator is a classical system that yields a classical estimate of a variable for the quantum plant. A robust quantum observer for uncertain quantum systems was constructed in Ref. \cite{NY}. On the other hand, here we build a robust complex classical estimator for uncertain linear quantum systems. A robust classical $H_\infty$ estimator for uncertain linear systems was presented in Ref. \cite{FDX}. Here, we develop a more general and complex $H_\infty$ filter for robust classical estimation of uncertain linear quantum systems. The solution to the $H_\infty$ estimation problem is obtained by means of two algebraic Riccati equations, upon converting the estimation problem to a scaled $H_\infty$ control problem.

The paper is structured as follows. We introduce the class of linear quantum systems considered in this paper in Section \ref{sec:lqs}. The complex Kalman filter is discussed in Section \ref{sec:kalman} for optimal classical estimation of linear quantum systems. Section \ref{sec:robust} then considers the robust $H_\infty$ estimation problem for uncertain linear quantum systems and presents our main result. Illustrative examples of our robust estimator for two different scenarios are provided in Sections \ref{sec:num_ex1} and \ref{sec:num_ex2}. Finally, Section \ref{sec:conc} concludes the paper with relevant remarks and possible future work.

\section{Linear Quantum Systems}\label{sec:lqs}
The class of linear quantum systems we consider here are described by the quantum stochastic differential equations (QSDEs) \cite{IRP2}:
\begin{equation}\label{eq:lqs_1}
\begin{split}
\left[\begin{array}{c}
da(t)\\
da(t)^{\#}
\end{array}\right] &= A \left[\begin{array}{c}
a(t)\\
a(t)^{\#}
\end{array}\right] dt + B \left[\begin{array}{c}
d\mathcal{A}(t)\\
d\mathcal{A}(t)^{\#}
\end{array}\right],\\
\left[\begin{array}{c}
d\mathcal{Y}(t)\\
d\mathcal{Y}(t)^{\#}
\end{array}\right] &= C \left[\begin{array}{c}
a(t)\\
a(t)^{\#}
\end{array}\right] dt + D \left[\begin{array}{c}
d\mathcal{A}(t)\\
d\mathcal{A}(t)^{\#}
\end{array}\right],
\end{split}
\end{equation}
where
\begin{equation}\label{eq:lqs_2}
\begin{split}
A &= \Delta(A_1,A_2), \qquad B = \Delta(B_1,B_2),\\
C &= \Delta(C_1,C_2), \qquad D = \Delta(D_1,D_2).
\end{split}
\end{equation}

Here, $a(t) = [a_1(t) \hdots a_n(t)]^T$ is a vector of annihilation operators. The adjoint of the operator $a_i$ is called a creation operator, denoted by $a_i^{*}$. The vector $\mathcal{A}(t) = [\mathcal{A}_1(t) \hdots \mathcal{A}_m(t)]^T$ represents a collection of external independent quantum field operators and the vector $\mathcal{Y}$ represents the corresponding vector of output field operators. Also, the notation $\Delta(A_1,A_2)$ denotes the matrix $\left[\begin{array}{cc} A_1 & A_2\\
A_2^{\#} & A_1^{\#}
\end{array}\right]$. Here, $A_1$, $A_2 \in \mathbb{C}^{n \times n}$, $B_1$, $B_2 \in \mathbb{C}^{n \times m}$, $C_1$, $C_2 \in \mathbb{C}^{m \times n}$, and $D_1$, $D_2 \in \mathbb{C}^{m \times m}$. Moreover, $^{\#}$ denotes the adjoint of a vector of operators or the complex conjugate of a complex matrix. Furthermore, $^\dagger$ denotes the adjoint transpose of a vector of operators or the complex conjugate transpose of a complex matrix.

\begin{definition}
\cite{IRP2} A complex linear quantum system of the form (\ref{eq:lqs_1}), (\ref{eq:lqs_2}) is said to be physically realizable if there exists a complex commutation matrix $\Theta = \Theta^\dagger$, a complex Hamiltonian matrix $M = M^\dagger$, and a coupling matrix $N$ such that
\begin{equation}\label{eq:theta}
\Theta = TJT^\dagger,
\end{equation}
where $J = \left[\begin{array}{cc}
I & 0\\
0 & -I
\end{array}\right]$, $T = \Delta(T_1,T_2)$ is non-singular, $M$ and $N$ are of the form
\begin{equation}
M = \Delta(M_1,M_2), \qquad N = \Delta(N_1,N_2),
\end{equation}
and
\begin{equation}
\begin{split}
A &= -\iota\Theta M - \frac{1}{2}\Theta N^\dagger JN,\\
B &= -\Theta N^\dagger J,\\
C &= N,\\
D &= I.
\end{split}
\end{equation}
\end{definition}

Here, $\iota = \sqrt{-1}$ is the imaginary unit, and the commutation matrix $\Theta$ satisfies the following commutation relation:
\begin{equation}\label{eq:comm_rel1}
\begin{split}
&\left[\left[\begin{array}{c}
a\\
a^{\#}
\end{array}\right], \left[\begin{array}{c}
a\\
a^{\#}
\end{array}\right]^\dagger\right]\\
&= \left[\begin{array}{c}
a\\
a^{\#}
\end{array}\right]\left[\begin{array}{c}
a\\
a^{\#}
\end{array}\right]^\dagger - \left(\left[\begin{array}{c}
a\\
a^{\#}
\end{array}\right]^\# \left[\begin{array}{c}
a\\
a^{\#}
\end{array}\right]^T\right)^T\\
&= \Theta .
\end{split}
\end{equation}

One can verify that $\Theta$ is a $2n \times 2n$ matrix, the elements of which are as follows, given $i,j = 1 \hdots n$:
\begin{align*}
\Theta_{ij} &= [a_i,a_j^{*}],\\
\Theta_{(n+i)(n+j)} &= [a_i^{*},a_j],\\
\Theta_{i(n+j)} &= [a_i,a_j],\\
\Theta_{(n+i)j} &= [a_i^{*},a_j^{*}],\\
\Theta_{ii} &= 1,\\
\Theta_{(n+i)(n+i)} &= -1,\\
\Theta_{i(n+i)} &= \Theta_{(n+i)i} = 0.
\end{align*}

The annihilation and creation operators can be used to construct the number operator $N = a^\dagger a$, the eigenstates of which form the orthonormal number (or Fock) states \cite{BR}:
\begin{equation*}
N|q\rangle = a^\dagger a|q\rangle = q|q\rangle , \quad q = 0,1,2,\hdots 
\end{equation*}

In particular, the state $|0\rangle$ is called the vacuum state \cite{BR}:
\begin{equation*}
a|0\rangle = 0.
\end{equation*}

The annihilation and creation operators have the properties of lowering and raising the number of a state respectively \cite{BR}:
\begin{align*}
a|q\rangle &= \sqrt{q}|q-1\rangle,\\
a^\dagger |q\rangle &= \sqrt{q+1}|q+1\rangle.
\end{align*}

\begin{theorem}\label{thm:phys_rlz}
\cite{IRP2} The linear quantum system (\ref{eq:lqs_1}), (\ref{eq:lqs_2}) is physically realizable if and only if there exists a complex matrix $\Theta = \Theta^\dagger$ such that $\Theta$ is of the form in (\ref{eq:theta}), and
\begin{equation}\label{eq:phys_rlz}
\begin{split}
A\Theta + \Theta A^\dagger + BJB^\dagger &= 0,\\
B &= -\Theta C^\dagger J,\\
D &= I.
\end{split}
\end{equation}
\end{theorem}

If the system (\ref{eq:lqs_1}) is physically realizable, then the matrices $M$ and $N$ define a complex open harmonic oscillator with a Hamiltonian operator
\[ \mathbf{H} = \frac{1}{2} \left[\begin{array}{cc}
a^\dagger & a^T
\end{array}\right] M \left[\begin{array}{c}
a\\
a^{\#}
\end{array}\right],\]
and a coupling operator
\[ \mathbf{L} = \left[\begin{array}{cc}
N_1 & N_2
\end{array}\right] \left[\begin{array}{c}
a\\
a^{\#}
\end{array}\right].\]

\section{Kalman Filter}\label{sec:kalman}
The schematic diagram of a classical estimation scheme is provided in Fig. \ref{fig:cls_scm}. We consider a quantum plant, which is a quantum system of the form (\ref{eq:lqs_1}), (\ref{eq:lqs_2}), defined as follows:
\begin{equation}\label{eq:plant}
\begin{split}
\left[\begin{array}{c}
da(t)\\
da(t)^{\#}
\end{array}\right] &= A \left[\begin{array}{c}
a(t)\\
a(t)^{\#}
\end{array}\right] dt + B \left[\begin{array}{c}
d\mathcal{A}(t)\\
d\mathcal{A}(t)^{\#} 
\end{array}\right],\\
\left[\begin{array}{c}
d\mathcal{Y}(t)\\
d\mathcal{Y}(t)^{\#}
\end{array}\right] &= C \left[\begin{array}{c}
a(t)\\
a(t)^{\#}
\end{array}\right] dt + D \left[\begin{array}{c}
d\mathcal{A}(t)\\
d\mathcal{A}(t)^{\#}
\end{array}\right],\\
z &= L\left[\begin{array}{c}
a(t)\\
a(t)^{\#} 
\end{array}\right].
\end{split}
\end{equation}

\begin{figure}
\centering
\includegraphics[width=\textwidth]{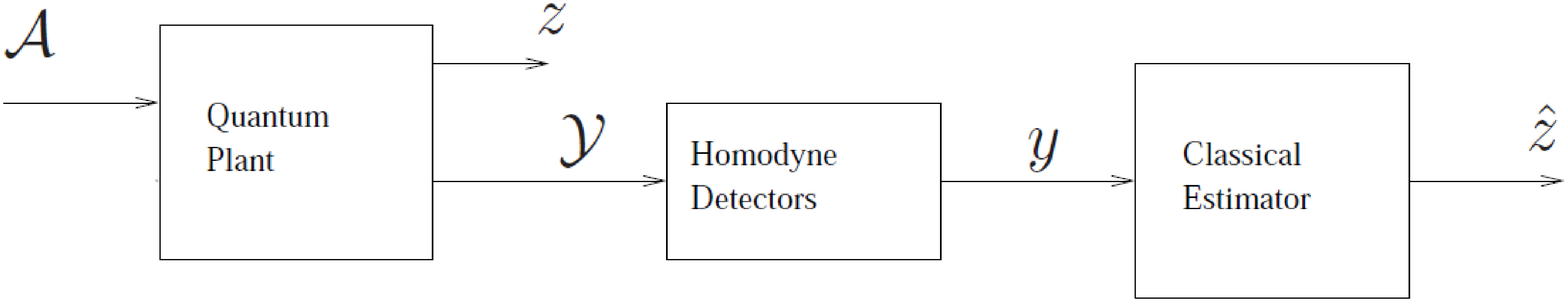}
\caption{Schematic diagram of classical estimation for a quantum plant.}
\label{fig:cls_scm}
\end{figure}

Here, $z$ denotes a scalar operator on the underlying Hilbert space and represents the quantity to be estimated. Also, $\mathcal{Y}(t)$ represents the vector of output fields of the plant, and $\mathcal{A}(t)$ represents a vector of quantum disturbances acting on the plant.

A quadrature of each component of the vector $\mathcal{Y}(t)$ is measured using homodyne detection to produce a corresponding classical signal $y_i$:
\begin{equation}\label{eq:class_hd}
\begin{split}
dy_1 &= \frac{e^{-\iota\theta_1}}{\sqrt{2}}d\mathcal{Y}_1 + \frac{e^{\iota\theta_1}}{\sqrt{2}}d\mathcal{Y}_1^{*},\\
&\vdots\\
dy_m &= \frac{e^{-\iota\theta_m}}{\sqrt{2}}d\mathcal{Y}_m + \frac{e^{\iota\theta_m}}{\sqrt{2}}d\mathcal{Y}_m^{*}.
\end{split}
\end{equation}

Here, the angles $\theta_1,\hdots,\theta_m$ determine the quadrature measured by each homodyne detector. The vector of classical signals $y = [y_1 \hdots y_m]^T$ is then used as the input to a classical estimator defined as follows:
\begin{equation}\label{eq:class_estimator}
\begin{split}
dx_e &= A_ex_edt + K_edy,\\
\hat{z} &= L_ex_e.
\end{split}
\end{equation}

For the sake of comparison, we will first consider the optimal estimation problem for quantum linear systems; see also \cite{NY,BHJ}. The optimal classical estimator is given by the standard (complex) Kalman filter defined for the system (\ref{eq:plant}), (\ref{eq:class_hd}). This optimal classical estimator is obtained from the solution to the algebraic Riccati equation:
\begin{equation}\label{eq:class_riccati}
\begin{split}
AP &+ PA^\dagger + BB^\dagger - (B + PC^\dagger)S^\dagger S(B + PC^\dagger)^\dagger =0,
\end{split}
\end{equation}
where
\begin{equation}
\begin{split}
S &= \left[\begin{array}{cc}
S_1 & S_2
\end{array}\right],\\
S_1 &= \left[\begin{array}{cccc}
\frac{e^{-\iota\theta_1}}{\sqrt{2}} & 0 & \hdots & 0\\
0 & \frac{e^{-\iota\theta_2}}{\sqrt{2}} & \hdots & 0\\
 & & \ddots & \\
 & & & \frac{e^{-\iota\theta_m}}{\sqrt{2}}
\end{array}\right],\\
S_2 &= \left[\begin{array}{cccc}
\frac{e^{\iota\theta_1}}{\sqrt{2}} & 0 & \hdots & 0\\
0 & \frac{e^{\iota\theta_2}}{\sqrt{2}} & \hdots & 0\\
 & & \ddots & \\
 & & & \frac{e^{\iota\theta_m}}{\sqrt{2}}
\end{array}\right].
\end{split}
\end{equation}

Here we have assumed that the quantum disturbance $\mathcal{A}$ is purely canonical, i.e. $d\mathcal{A}d\mathcal{A}^\dagger = Idt$ and hence $D=I$.

Then, the corresponding optimal classical estimator (\ref{eq:class_estimator}) is defined by the equations:
\begin{equation}\label{eq:cls_sys_est}
\begin{split}
A_e &= A - K_eSC,\\
K_e &= (B + PC^\dagger)S^\dagger,\\
L_e &= L.
\end{split}
\end{equation}

\section{Robust $H_\infty$ Filter}\label{sec:robust}

Corresponding to the system described by (\ref{eq:plant}), (\ref{eq:class_hd}), we define our uncertain system modelled as follows:
\begin{equation}\label{eq:uncertain1}
\begin{split}
(\Sigma_1): \dot{x}(t) &= [A+\Delta A(t)]x(t) + [B+\Delta B(t)]w(t),\\
z(t) &= Lx(t),\\
y'(t) &= S[C+\Delta C(t)]x(t) + SDw(t),
\end{split}
\end{equation}
where $x(t) := \left[\begin{array}{c}
a(t)\\
a(t)^{\#}
\end{array}\right]$ is the state, $w(t)$ is the disturbance input, $z(t)$ is a linear combination of the state variables to be estimated, $y'(t)$ is the measured output, $L \in \mathbb{C}^{p \times 2n}$, $SC \in \mathbb{C}^{m \times 2n}$, $SD \in \mathbb{C}^{m \times 2m}$, and $\Delta A(\cdot)$, $\Delta B(\cdot)$ and $\Delta C(\cdot)$ denote the time-varying parameter uncertainties. These uncertainties are in the following structure
\begin{equation}\label{eq:unc_pars}
\begin{split}
\left[\begin{array}{c}
\Delta A(t)\\
\Delta C(t)
\end{array}\right] &= \left[\begin{array}{c}
H_1\\
H_3
\end{array}\right]F_1(t)E,\\
\Delta B(t) &= H_2F_2(t)G,
\end{split}
\end{equation}
where $H_1$, $H_2$, $H_3$, $E$ and $G$ are known complex constant matrices with appropriate dimensions, and the unknown matrix functions $F_1(\cdot)$ and $F_2(\cdot)$ satisfy the following:
\begin{equation}\label{eq:unc_constraint}
\begin{split}
F_1^\dagger(t)F_1(t) &\leq I, \qquad \forall t,\\
F_2^\dagger(t)F_2(t) &\leq I, \qquad \forall t.
\end{split}
\end{equation}

Note that for the system (\ref{eq:uncertain1}) to be physically realizable, the following constraints are required to be satisfied by the uncertainties:
\begin{equation}\label{eq:phys_rlz_unc}
\begin{split}
\Delta A\Theta + \Theta\Delta A^\dagger + BJ\Delta B^\dagger + \Delta BJB^\dagger + \Delta BJ\Delta B^\dagger &= 0,\\
\Delta B &= -\Theta\Delta C^\dagger J.
\end{split}
\end{equation}

The robust $H_\infty$ estimation problem for the uncertain system (\ref{eq:uncertain1}) can be converted into a scaled $H_\infty$ control problem, as described in Ref. \cite{FDX}, by introducing the following parameterized linear time-invariant system corresponding to system (\ref{eq:uncertain1}):
\begin{equation}\label{eq:uncertain2}
\begin{split}
(\Sigma_2): \dot{x}(t) &= Ax(t) + \left[\begin{array}{ccc}
B & \frac{\gamma}{\epsilon_1}H_1 & \frac{\gamma}{\epsilon_2}H_2
\end{array}\right]\tilde{w}(t),\\
\tilde{z}(t) &= \left[\begin{array}{c}
\epsilon_1E\\
0\\
L
\end{array}\right]x(t) + \left[\begin{array}{ccc}
0 & 0 & 0\\
\epsilon_2G & 0 & 0\\
0 & 0 & 0
\end{array}\right]\tilde{w}(t) + \left[\begin{array}{c}
0\\
0\\
-I
\end{array}\right]u(t),\\
y'(t) &= SCx(t) + \left[\begin{array}{ccc}
SD & \frac{\gamma}{\epsilon_1}SH_3 & 0
\end{array}\right]\tilde{w}(t).
\end{split}
\end{equation}

Here, $u(t)$ is the control input, $\tilde{z}(t)$ is the controlled output, $\epsilon_1$, $\epsilon_2 > 0$ are suitably chosen scaling parameters and $\gamma > 0$ is the desired level of disturbance attenuation for the robust $H_\infty$ estimation problem. We also have the augmented disturbance
\[ \tilde{w}(t) := \left[\begin{array}{c}
w(t)\\
\frac{\epsilon_1}{\gamma}\eta(t)\\
\frac{\epsilon_2}{\gamma}\xi(t)
\end{array}\right],\]
where
\[ \eta(t) := F_1(t)Ex(t),\]
and 
\[ \xi(t) := F_2(t)Gw(t).\]

The following assumptions are made for the system (\ref{eq:uncertain2}):
\begin{assumption}
\begin{itemize}
\item[]
\item[A1.] The system matrix $A$ is stable.
\item[A2.] $\epsilon_2^2G^\dagger G < I$.
\item[A3.] $\left[\begin{array}{cc}
SD & SH_3
\end{array}\right]$ is of full row rank.
\item[A4.] rank $\left[\begin{array}{cc}
A-\iota\omega I & B\\
C & D
\end{array}\right] = 2n+2m$, \qquad $\forall\omega\geq 0$.
\end{itemize}
\end{assumption}

\begin{remark}
The assumption A1 is required so that the $H_\infty$ norm for the combined plant-estimator system is finite. The remaining assumptions are technical assumptions arising in $H_\infty$ control theory which are required in order to obtain a solution using Riccati equations.
\end{remark}

A complete solution to the robust $H_\infty$ estimation problem is then provided below.

\begin{theorem}\label{thm:h_infinity}
Consider the robust $H_\infty$ estimation problem for the uncertain system (\ref{eq:uncertain1}) converted to a scaled $H_\infty$ control problem for the system (\ref{eq:uncertain2}), satisfying the assumptions A1 to A4. Given a prescribed level of disturbance attenuation $\gamma > 0$, the robust $H_\infty$ estimation problem for the uncertain system (\ref{eq:uncertain1}) is solvable if for some $\epsilon_1$, $\epsilon_2 > 0$, the following conditions are satisfied:

(a) There exists a stabilising solution $X = X^\dagger \geq 0$ to the algebraic Riccati equation:
\begin{equation}\label{eq:robust_riccati1}
\overline{A}^\dagger X+X\overline{A}+X(\gamma^{-2} \overline{B}_1\overline{B}_1^\dagger)X + \overline{C}_1^\dagger (I-\overline{D}_{12}\overline{E}_1^{-1}\overline{D}_{12}^\dagger)\overline{C}_1 = 0.
\end{equation}

(b) There exists a stabilising solution $Y = Y^\dagger \geq 0$ to the algebraic Riccati equation:
\begin{equation}\label{eq:robust_riccati2}
\begin{split}
\overline{A}Y&+Y\overline{A}^\dagger +Y\overline{C}_1^\dagger\overline{C}_1Y+ \gamma^{-2}\overline{B}_1\overline{B}_1^\dagger \\ &-(\gamma^{-1}\overline{B}_1\overline{D}_{21}^\dagger+\gamma Y\overline{C}_2^\dagger) \overline{S}^\dagger\overline{E}_2^{-1}\overline{S}(\gamma^{-1}\overline{B}_1\overline{D}_{21}^\dagger+\gamma Y\overline{C}_2^\dagger)^\dagger = 0.
\end{split}
\end{equation}

(c) $I-\gamma^{-2}XY > 0$.

Here, we have
\vspace*{-2mm}
\begin{equation}\label{eq:final_parameters1}
\begin{split}
\overline{A} &= A,\\
\overline{B}_1 &= \left[\begin{array}{ccc}
B(I-\epsilon_2^2G^\dagger G)^{-1/2} & \frac{\gamma}{\epsilon_1}H_1 & \frac{\gamma}{\epsilon_2}H_2
\end{array}\right],\\
\overline{C}_1 &= \left[\begin{array}{c}
\epsilon_1E\\
0\\
L
\end{array}\right],\\
\overline{D}_{12} &= \left[\begin{array}{c}
0\\
0\\
-I
\end{array}\right],\\
\overline{C}_2 &= C,\\
\overline{D}_{21} &= \left[\begin{array}{ccc}
D(I-\epsilon_2^2G^\dagger G)^{-1/2} & \frac{\gamma}{\epsilon_1}H_3 & 0
\end{array}\right],\\
\overline{S} &= S,
\end{split}
\end{equation}
and
\begin{equation}\label{eq:final_parameters3}
\begin{split}
\overline{E}_1 &= \overline{D}_{12}^\dagger\overline{D}_{12} = I,\\
\overline{E}_2 &= \overline{S}\overline{D}_{21}\overline{D}_{21}^\dagger \overline{S}^\dagger = SD(I-\epsilon_2^2G^\dagger G)^{-1}D^\dagger S^\dagger + \frac{\gamma^2}{\epsilon_1^2}SH_3H_3^\dagger S^\dagger .
\end{split}
\end{equation}

When conditions (a)-(c) are satisfied, a suitable estimator is given by:
\begin{equation}\label{eq:robust_estimator}
\begin{split}
\dot{\hat{x}}(t) &= A_K\hat{x}(t) + B_Ky'(t),\\
\hat{z}(t) &= C_K\hat{x}(t),
\end{split}
\end{equation}
where
\vspace*{-2mm}
\begin{equation}
\begin{split}
A_K &= \overline{A} - B_K\overline{S}\overline{C}_2 + \gamma^{-2}(\overline{B}_1-B_K\overline{S}\overline{D}_{21})\overline{B}_1^\dagger X,\\
B_K &= \gamma^2(I-YX)^{-1}(Y\overline{C}_2^\dagger\overline{S}^\dagger + \gamma^{-2} \overline{B}_1\overline{D}_{21}^\dagger\overline{S}^\dagger)\overline{E}_2^{-1},\\
C_K &= -\overline{E}_1^{-1}\overline{D}_{12}^\dagger\overline{C}_1.
\end{split}
\end{equation}
\end{theorem}

\begin{proof}
The system (\ref{eq:uncertain2}) is of the following form:
\begin{equation}\label{eq:uncertain3}
\begin{split}
(\Sigma_2): \dot{x}(t) &= Ax(t) + B_1\tilde{w}(t) + B_2u(t),\\
\tilde{z}(t) &= C_1x(t) + D_{11}\tilde{w}(t) + D_{12}u(t),\\
y'(t) &= SC_2x(t) + SD_{21}\tilde{w}(t) + SD_{22}u(t),
\end{split}
\end{equation}
where 
\begin{equation}
\begin{split}
B_1 &= \left[\begin{array}{ccc}
B & \frac{\gamma}{\epsilon_1}H_1 & \frac{\gamma}{\epsilon_2}H_2
\end{array}\right],\\
B_2 &= 0,\\
C_1 &= \left[\begin{array}{c}
\epsilon_1E\\
0\\
L
\end{array}\right],\\
D_{11} &= \left[\begin{array}{ccc}
0 & 0 & 0\\
\epsilon_2G & 0 & 0\\
0 & 0 & 0
\end{array}\right],\\
D_{12} &= \left[\begin{array}{c}
0\\
0\\
-I
\end{array}\right],\\
C_2 &= C,\\
D_{21} &= \left[\begin{array}{ccc}
D & \frac{\gamma}{\epsilon_1}H_3 & 0
\end{array}\right],\\
D_{22} &= 0.
\end{split}
\end{equation}

We will use the results from Ref. \cite{PAJ} to solve the above $H_\infty$ control problem. However, Ref. \cite{PAJ} requires the matrix $D_{11}$ to be zero. The system (\ref{eq:uncertain2}) with non-zero $D_{11}$ can be converted to an equivalent system having $D_{11} = 0$ using the \emph{loop-shifting} technique, as outlined in Ref. \cite{ZDG}.

We define
\begin{equation}
\tilde{D}_{11} := \left[\begin{array}{cc}
D_{1111} & D_{1112}\\
D_{1121} & D_{1122}+D_\infty
\end{array}\right],
\end{equation}
where we let $D_\infty = -D_{1122}-D_{1121}(I - D_{1111}^\dagger D_{1111})^{-1}D_{1111}^\dagger D_{1112}$. Here, for our system (\ref{eq:uncertain2}), we take $D_{1111} = \left[\begin{array}{cc}
0 & 0\\
\epsilon_2G & 0\\
\end{array}\right]$, $D_{1121} = \left[\begin{array}{cc}
0 & 0\\
\end{array}\right]$, $D_{1112} = \left[\begin{array}{c}
0\\
0
\end{array}\right]$, and $D_{1122} = 0$. Note that $\left|\left|\tilde{D}_{11}\right|\right|<1$ follows from Assumption A2. Then, we get $D_\infty = 0$. Thus, we have $\tilde{D}_{11} = D_{11}$.

Hence, the $H_\infty$ control problem in (\ref{eq:uncertain3}) takes the following form:
\begin{equation}\label{eq:uncertain4}
\begin{split}
(\Sigma_3): \dot{x}(t) &= \tilde{A}x(t) + \tilde{B}_1\tilde{w}(t) + \tilde{B}_2u(t),\\
\tilde{z}(t) &= \tilde{C}_1x(t) + \tilde{D}_{11}\tilde{w}(t) + \tilde{D}_{12}u(t),\\
y'(t) &= \tilde{S}\tilde{C}_2x(t) + \tilde{S}\tilde{D}_{21}\tilde{w}(t) + \tilde{S}\tilde{D}_{22}u(t),
\end{split}
\end{equation}
where
\begin{equation}
\begin{split}
\tilde{A} &= A+B_2D_\infty C_2 = A,\\
\tilde{B}_1 &= B_1+B_2D_\infty D_{21} = B_1,\\
\tilde{B}_2 &= B_2 = 0,\\
\tilde{C}_1 &= C_1 + D_{12}D_\infty C_2 = C_1,\\
\tilde{D}_{12} &= D_{12},\\
\tilde{C}_2 &= C_2,\\
\tilde{D}_{21} &= D_{21},\\
\tilde{D}_{22} &= D_{22} = 0,\\
\tilde{S} &= S.
\end{split}
\end{equation}

The $H_\infty$ control problem equivalent to the above system (\ref{eq:uncertain4}) is \cite{ZDG}:
\begin{equation}\label{eq:uncertain5}
\begin{split}
(\Sigma_4): \dot{x}(t) &= \overline{A}x(t) + \overline{B}_1\tilde{w}(t) + \overline{B}_2u(t),\\
\tilde{z}(t) &= \overline{C}_1x(t) + \overline{D}_{11}\tilde{w}(t) + \overline{D}_{12}u(t),\\
y'(t) &= \overline{S}\overline{C}_2x(t) + \overline{S}\overline{D}_{21}\tilde{w}(t) + \overline{S}\overline{D}_{22}u(t),
\end{split}
\end{equation}

Here,
\begin{equation}\label{eq:final_parameters4}
\begin{split}
\overline{A} &= \tilde{A}+\tilde{B}_1R_1^{-1}\tilde{D}_{11}^\dagger \tilde{C}_1,\\
\overline{B}_1 &= \tilde{B}_1R_1^{-1/2},\\
\overline{B}_2 &= \tilde{B}_2 + \tilde{B}_1R_1^{-1}\tilde{D}_{11}^\dagger\tilde{D}_{12},\\
\overline{C}_1 &= \tilde{R}_1^{-1/2}\tilde{C}_1,\\
\overline{D}_{11} &= 0,\\
\overline{D}_{12} &= \tilde{R}_1^{-1/2}\tilde{D}_{12},\\
\overline{C}_2 &= \tilde{C}_2 + \tilde{D}_{21}R_1^{-1}\tilde{D}_{11}^\dagger \tilde{C}_1,\\
\overline{D}_{21} &= \tilde{D}_{21}R_1^{-1/2},\\
\overline{D}_{22} &= \tilde{D}_{21}R_1^{-1}\tilde{D}_{11}^\dagger\tilde{D}_{12},\\
\overline{S} &= S,
\end{split}
\end{equation}
where
\begin{equation}
\begin{split}
R_1 &= I - \tilde{D}_{11}^\dagger\tilde{D}_{11} = \left[\begin{array}{ccc}
I-\epsilon_2^2G^\dagger G & 0 & 0\\
0 & I & 0\\
0 & 0 & I
\end{array}\right],\\
\tilde{R}_1 &= I - \tilde{D}_{11}\tilde{D}_{11}^\dagger = \left[\begin{array}{ccc}
I & 0 & 0\\
0 & I-\epsilon_2^2GG^\dagger & 0\\
0 & 0 & I
\end{array}\right].
\end{split}
\end{equation}

One can verify that (\ref{eq:final_parameters4}) yields (\ref{eq:final_parameters1}), $\overline{B}_2 = 0$ and $\overline{D}_{22} = 0$. Note that from assumptions A1 to A4 and (\ref{eq:final_parameters3}), it follows that we have:
\begin{itemize}
\item $\overline{E}_1 > 0$.
\item $\overline{E}_2 > 0$.
\item rank $\left[\begin{array}{cc}
\overline{A}-\iota\omega I & \overline{B}_2\\
\overline{C}_1 & \overline{D}_{12}
\end{array}\right] = 2n + p$ for all $\omega \geq 0$.
\item rank $\left[\begin{array}{cc}
\overline{A}-\iota\omega I & \overline{B}_1\\
\overline{C}_2 & \overline{D}_{21}
\end{array}\right] = 2n + m$ for all $\omega \geq 0$.
\end{itemize}

Hence, the $H_\infty$ control problem for the system (\ref{eq:uncertain5}) can be solved using the results in Ref. \cite{PAJ}. However, Ref. \cite{PAJ} assumed $\gamma = 1$. The results from that paper may be generalised for different values of $\gamma > 0$, simply by scaling the coefficients of the disturbance $\tilde{w}(t)$ in (\ref{eq:uncertain5}), viz. $\overline{B}_1$ and $\overline{D}_{21}$ (note $\overline{D}_{11}$ = 0), by $\gamma^{-1}$. Note that this also has an effect on $\overline{E}_2$, which is scaled by $\gamma^{-2}$. This yields conditions (a)-(c) of the theorem, which are required to be satisfied by the system (\ref{eq:uncertain5}), such that a suitable estimator is given by (\ref{eq:robust_estimator}).
\end{proof}

The transfer function of the robust estimator (\ref{eq:robust_estimator}) can be obtained to be:
\begin{equation}\label{eq:rob_filter_tf}
G_K(s) := \frac{\hat{z}(s)}{y'(s)} = C_K(sI-A_K)^{-1}B_K.
\end{equation}

The estimation error is given as:
\begin{equation}
e(t) := \hat{z}(t)-z(t) = C_K\hat{x}(t)-Lx(t).
\end{equation}

Then, the disturbance-to-error transfer function may be obtained as:
\begin{equation}\label{eq:error_spectrum}
\tilde{G}_{we}(s) := \frac{e(s)}{w(s)} = \left[\begin{array}{cc}
-L & C_K
\end{array}\right]\left(sI - \left[\begin{array}{cc}
A+\Delta A & 0\\
B_KS(C+\Delta C) & A_K
\end{array}\right]\right)^{-1}\left[\begin{array}{c}
B+\Delta B\\
B_KSD
\end{array}\right].
\end{equation}

We are interested in the disturbance $\mathcal{A}$ to error $e$ transfer function, which is simply the first component of the matrix transfer function $\tilde{G}_{we}(s)$. The other component is the disturbance $\mathcal{A}^{\#}$ to error $e$ transfer function, which we shall ignore.

\begin{remark}
Note that the fact that the plant is a quantum system that will be physically realizable \emph{restricts} the class of plants under consideration, when compared to the case when the plant is a classical system as in Ref. \cite{FDX}, owing to the conditions in (\ref{eq:phys_rlz}) (and also (\ref{eq:phys_rlz_unc})) required to be satisfied by the system matrices of the quantum plant (uncertain plant).
\end{remark}

\section{Numerical Example 1}\label{sec:num_ex1}
An example of a linear quantum system from quantum optics is a linearized dynamic squeezer. This corresponds to an optical cavity with a non-linear active medium inside. Let us consider the quantum plant to be a linearized dynamic squeezer, described by the QSDEs \cite{RPH}:

\begin{equation}\label{eq:sqz_plant}
\begin{split}
\left[\begin{array}{c}
da(t)\\
da(t)^{*}
\end{array}\right] &= \left[\begin{array}{cc}
-\frac{\beta}{2} & -\chi\\
-\chi^{*} & -\frac{\beta}{2}
\end{array}\right] \left[\begin{array}{c}
a(t)\\
a(t)^{*}
\end{array}\right] dt - \sqrt{\kappa} \left[\begin{array}{c}
d\mathcal{A}(t)\\
d\mathcal{A}(t)^{*}
\end{array}\right],\\
\left[\begin{array}{c}
d\mathcal{Y}(t)\\
d\mathcal{Y}(t)^{*}
\end{array}\right] &= \sqrt{\kappa} \left[\begin{array}{c}
a(t)\\
a(t)^{*}
\end{array}\right] dt + \left[\begin{array}{c}
d\mathcal{A}(t)\\
d\mathcal{A}(t)^{*}
\end{array}\right],\\
z(t) &= \left[\begin{array}{cc}
0.1 & -0.1
\end{array}\right] \left[\begin{array}{c}
a(t)\\
a(t)^{*}
\end{array}\right],
\end{split}
\end{equation}
where $\beta > 0$ is the overall cavity loss, $\kappa > 0$ determines the loss arising from the cavity mirrors, $\chi \in \mathbb{C}$ quantifies the size of the non-linearity of the active medium, and $a$ is a single annihilation operator of the cavity mode.

Here, we choose $\beta = 4$, $\kappa = 4$, and $\chi = -0.5$. Then, the above quantum system is physically realizable, since we have $\beta = \kappa$. Moreover, we fix the homodyne detection angle at $90^{\circ}$. We, thus, have the following:
\begin{equation}\label{eq:sqz_plant_matrices}
\begin{split}
A &= \left[\begin{array}{cc}
-2 & 0.5\\
0.5 & -2
\end{array}\right], \,
B = \left[\begin{array}{cc}
-2 & 0\\
0 & -2
\end{array}\right], \,
C = \left[\begin{array}{cc}
2 & 0\\
0 & 2
\end{array}\right],\\
D &= \left[\begin{array}{cc}
1 & 0\\
0 & 1
\end{array}\right] = I, \,
L = \left[\begin{array}{cc}
0.1 & -0.1
\end{array}\right], \,
S = \left[\begin{array}{cc}
\frac{e^{-\iota 90^{\circ}}}{\sqrt{2}} & \frac{e^{\iota 90^{\circ}}}{\sqrt{2}}
\end{array}\right].
\end{split}
\end{equation}

We introduce uncertainty in the parameter $\alpha := \sqrt{\kappa}$ as follows: $\alpha \rightarrow \alpha+\mu\delta\alpha$, where $|\delta| \leq 1$ is an uncertain parameter and $\mu \in [0,1)$ determines the level of uncertainty. Then, we will have the following:
\begin{equation}\label{eq:sqz_plant_unc1}
\begin{split}
\Delta A &= \left[\begin{array}{cc}
-\mu\delta\alpha^2-\frac{\mu^2\delta^2\alpha^2}{2} & 0\\
0 & -\mu\delta\alpha^2-\frac{\mu^2\delta^2\alpha^2}{2}
\end{array}\right],\\
\Delta B &= \left[\begin{array}{cc}
-\mu\delta\alpha & 0\\
0 & -\mu\delta\alpha
\end{array}\right],\\
\Delta C &= \left[\begin{array}{cc}
\mu\delta\alpha & 0\\
0 & \mu\delta\alpha
\end{array}\right].
\end{split}
\end{equation}

Then, we define the relevant matrices in (\ref{eq:unc_pars}) as follows:
\begin{equation}\label{eq:sqz_plant_unc2}
\begin{split}
F_1(t) &= \left[\begin{array}{cccc}
\delta & 0 & 0 & 0\\
0 & \delta & 0 & 0\\
0 & 0 & \delta^2 & 0\\
0 & 0 & 0 & \delta^2
\end{array}\right],\\
F_2(t) &= \left[\begin{array}{cc}
\delta & 0\\
0 & \delta
\end{array}\right],\\
E &= \left[\begin{array}{cc}
-\frac{1}{2} & 0\\
0 & -\frac{1}{2}\\
-\frac{1}{2} & 0\\
0 & -\frac{1}{2}
\end{array}\right],\\
G &= \left[\begin{array}{cc}
1 & 0\\
0 & 1
\end{array}\right],\\
H_1 &= \left[\begin{array}{cccc}
2\mu\alpha^2 & 0 & \mu^2\alpha^2 & 0\\
0 & 2\mu\alpha^2 & 0 & \mu^2\alpha^2
\end{array}\right],\\
H_2 &= \left[\begin{array}{cc}
-\mu\alpha & 0\\
0 & -\mu\alpha
\end{array}\right],\\
H_3 &= \left[\begin{array}{cccc}
-2\mu\alpha & 0 & 0 & 0\\
0 & -2\mu\alpha & 0 & 0
\end{array}\right].
\end{split}
\end{equation}

One can then verify that we have $\Delta A = H_1F_1(t)E$, $\Delta B = H_2F_2(t)G$ and $\Delta C = H_3F_1(t)E$, as required in (\ref{eq:unc_pars}). We set the uncertainty level to $\mu = 0.1$. Moreover, in our simulations, we choose a fixed value of $\delta = 1$.

We now solve the associated $H_\infty$ estimation problem using the Riccati equation approach described in the previous section. We choose the desired disturbance attenuation level to be $\gamma = 0.65$. Then, the scaling parameters are suitably chosen to be $\epsilon_1 = 0.2$ and $\epsilon_2 = 0.6$. A robust $H_\infty$ estimator is obtained as in (\ref{eq:robust_estimator}) with the following parameters:
\begin{equation}
\begin{split}
A_K &= \left[\begin{array}{cc}
0.1905 & -1.4676\\
-1.4676 & 0.1905
\end{array}\right],\\
B_K &= \left[\begin{array}{c}
-1.4717\iota\\
1.4717\iota
\end{array}\right],\\
C_K &= \left[\begin{array}{cc}
0.1 & -0.1
\end{array}\right].
\end{split}
\end{equation}

The transfer function (\ref{eq:rob_filter_tf}) of the estimator is obtained to be:
\begin{equation}
G_K(s) = \frac{-0.2943\iota s-0.3759\iota}{s^2-0.381s-2.118}.
\end{equation}

Fig. \ref{fig:errors_spectra1} shows the disturbance to error transfer function (\ref{eq:error_spectrum}) of the above robust filter. It also shows those of the optimal $H_\infty$ filter and the Kalman filter for the uncertain system (\ref{eq:uncertain1}) with $\delta = 1$ in our example for comparison, which corresponds to the uncertain parameter taking on its maximum value. The optimal Kalman and $H_\infty$ filters are built for the nominal system of (\ref{eq:uncertain1}).

\begin{figure}[!b]
\centering
\includegraphics[width=\textwidth]{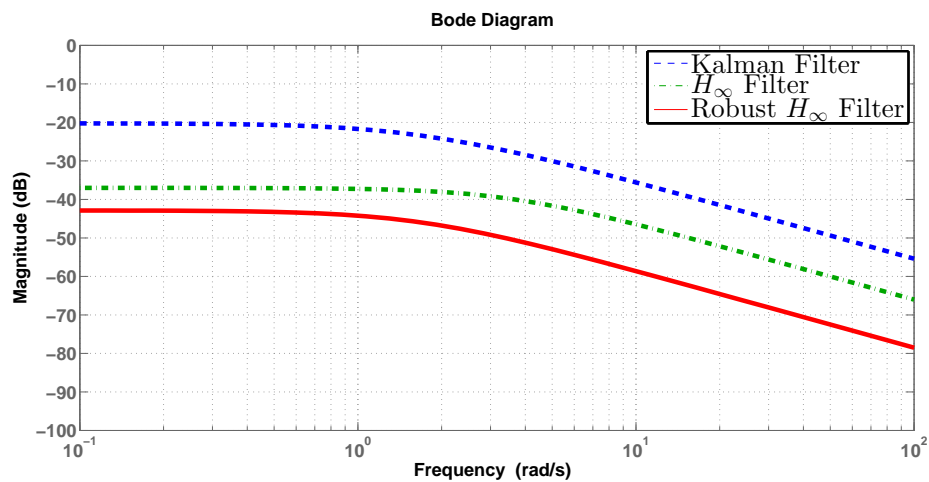}
\caption{Example 1: Disturbance to estimation error transfer functions for various estimators.}
\label{fig:errors_spectra1}
\end{figure}

Clearly, the robust $H_\infty$ filter provides for preferable disturbance attenuation and performance than the standard optimal $H_\infty$ filter or the Kalman filter for this value of the uncertain parameter $\delta$.

\section{Numerical Example 2}\label{sec:num_ex2}
In this section, we consider another numerical example involving the dynamic squeezer plant (\ref{eq:sqz_plant}), (\ref{eq:sqz_plant_matrices}). Here, we do not consider uncertainty in $\beta$ or $\kappa$. Instead, we introduce uncertainty in the squeezing parameter $\chi$ as follows: $\chi \rightarrow \chi + \mu\delta\chi$, where again $|\delta| \leq 1$ is an uncertain parameter and $0 \leq \mu < 1$ is the extent of uncertainty. Then, we will have the following:
\begin{equation}\label{eq:sqz_plant_unc3}
\begin{split}
\Delta A &= \left[\begin{array}{cc}
0 & -\mu\delta\chi\\
-\mu\delta\chi & 0
\end{array}\right],\\
\Delta B &= \left[\begin{array}{cc}
0 & 0\\
0 & 0
\end{array}\right],\\
\Delta C &= \left[\begin{array}{cc}
0 & 0\\
0 & 0
\end{array}\right].
\end{split}
\end{equation}

Note that the results from Ref. \cite{FDX} are enough to construct the robust estimator for this scenario. When $\Delta B = 0$ as in this case, $\epsilon_2$ has no impact on the estimator.

Here, the relevant matrices may be defined as follows:
\begin{equation}\label{eq:sqz_plant_unc4}
\begin{split}
F_1(t) &= \left[\begin{array}{cc}
\delta & 0\\
0 & \delta\\
\end{array}\right],\\
F_2(t) &= 0,\\
E &= \left[\begin{array}{cc}
\chi & 0\\
0 & \chi\\
\end{array}\right],\\
G &= 0,\\
H_1 &= \left[\begin{array}{cc}
0 & -\mu\\
-\mu & 0
\end{array}\right],\\
H_2 &= 0,\\
H_3 &= \left[\begin{array}{cc}
0 & 0\\
0 & 0
\end{array}\right].
\end{split}
\end{equation}

We fix the uncertainty level at $\mu = 0.1$. We then solve the associated $H_\infty$ estimation problem using our Riccati equation approach. We again choose the desired disturbance attenuation level to be $\gamma = 0.65$. The scaling parameters are chosen to be $\epsilon_1 = 0.7$ and $\epsilon_2 = 1$. A robust $H_\infty$ estimator is obtained as in (\ref{eq:robust_estimator}) with the following parameters:
\begin{equation}
\begin{split}
A_K &= \left[\begin{array}{cc}
0.3231 & -1.3660\\
-1.3660 & 0.3231
\end{array}\right],\\
B_K &= \left[\begin{array}{c}
-1.4852\iota\\
1.4852\iota
\end{array}\right],\\
C_K &= \left[\begin{array}{cc}
0.1 & -0.1
\end{array}\right].
\end{split}
\end{equation}

The transfer function (\ref{eq:rob_filter_tf}) of the estimator is obtained to be:
\begin{equation}
G_K(s) = \frac{-0.297\iota s-0.3098\iota}{s^2-0.6461s-1.762}.
\end{equation}

The disturbance to error transfer function (\ref{eq:error_spectrum}) of the robust filter in this example is shown in Fig. \ref{fig:errors_spectra_chi}. We also plot those of the optimal $H_\infty$ filter and the Kalman filter for the uncertain system (\ref{eq:uncertain1}) with $\delta = 1$ here for comparison. The optimal $H_\infty$ and Kalman filters are constructed for the nominal system of (\ref{eq:uncertain1}).

\begin{figure}
\centering
\includegraphics[width=\textwidth]{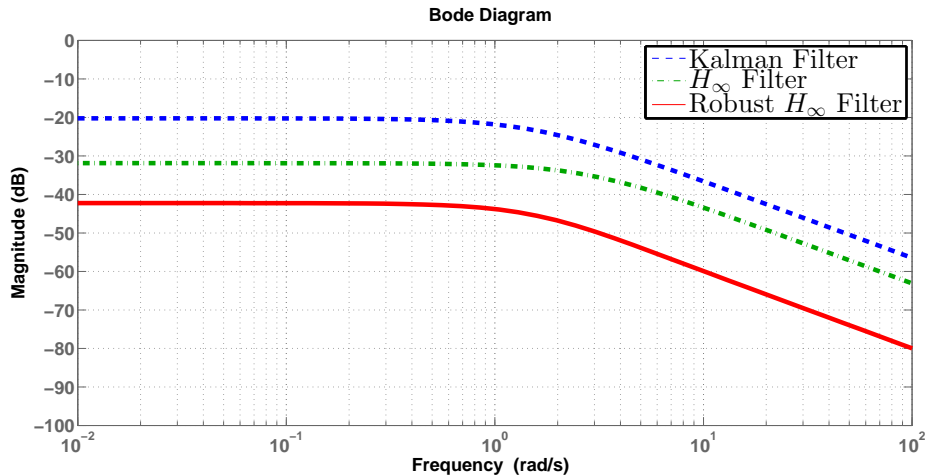}
\caption{Example 2: Disturbance to estimation error transfer functions of various estimators.}
\label{fig:errors_spectra_chi}
\end{figure}

Clearly, the robust $H_\infty$ filter again provides for preferable disturbance attenuation and performance compared with the standard optimal $H_\infty$ filter or the Kalman filter for this value of the uncertain parameter $\delta$.

\section{Conclusion}\label{sec:conc}
This paper considered the problem of robust $H_\infty$ estimation for uncertain linear quantum systems. The estimator is a classical filter, that produces a classical estimate of a variable of the quantum plant. The $H_\infty$ estimation problem is solved by converting it to a scaled $H_\infty$ control problem. The solution is obtained in the form of two algebraic complex Riccati equations. We have illustrated the results obtained by means of some numerical examples involving dynamic optical squeezers. As part of future work, the robust $H_\infty$ estimator constructed here could be applied in studying robust coherent-classical estimation. For example, it might be interesting to explore if and when a robust $H_\infty$ coherent-classical estimator (with and/or without coherent feedback) can provide better estimation precision than the robust $H_\infty$ purely-classical estimator considered in this paper. Such a comparison of optimal estimators was presented in Ref. \cite{RPH1}.

\ack The first author would like to thank Dr. Obaid Ur Rehman for useful discussion related to this work.

\bibliography{rhebib.bib}

\end{document}